\documentclass{article}

\usepackage{fancybox}  % for Ovalbox
\usepackage{amsthm}    % for qed

\newcommand\x{{\bf x}}
\newcommand\bff{{\bf f}}
\newcommand\bfg{{\bf g}}
\newcommand\bfh{{\bf h}}
\newcommand\zero{{\bf 0}}

\newtheorem{theorem}{Theorem}[section]
\newtheorem{example}[theorem]{Example}
\newtheorem{corollary}[theorem]{Corollary}

\title{A Blackbox Polynomial System Solver \\
on Parallel Shared Memory Computers\thanks{This material is based upon work 
supported by the National Science Foundation under Grant No. 1440534.}}

\author{Jan Verschelde \\
University of Illinois at Chicago \\
Department of Mathematics, Statistics, and Computer Science \\
851 S. Morgan Street (m/c 249), Chicago, IL 60607-7045, USA \\
{\tt janv@uic.edu, http://www.math.uic.edu/$\sim$jan}}

\begin{document}

\maketitle

\begin{abstract}
A numerical irreducible decomposition for a polynomial system
provides representations for the irreducible factors of all 
positive dimensional solution sets of the system,
separated from its isolated solutions.
Homotopy continuation methods are applied
to compute a numerical irreducible decomposition.
Load balancing and pipelining are techniques in
a parallel implementation on a computer with multicore processors.
The application of the parallel algorithms is illustrated
on solving the cyclic $n$-roots problems,
in particular for $n = 8, 9$, and~12.

\noindent {\bf Keywords and phrases.}
Homotopy continuation,
numerical irreducible decomposition,
mathematical software, multitasking,
pipelining, polyhedral homotopies, polynomial system,
shared memory parallel computing.

\end{abstract}

\section{Introduction}

Almost all computers have multicore processors enabling
the simultaneous execution of instructions in an algorithm.
The algorithms considered in this paper are applied to solve
a polynomial system.  Parallel algorithms can often
deliver significant speedups on computers with multicore processors.

A blackbox solver implies a fixed selection of algorithms,
run with default settings of options and tolerances.
The selected methods are homotopy continuation methods to compute
a numerical irreducible decomposition of the solution set of
a polynomial system.  
As the solution paths defined by a polynomial homotopy can be 
tracked independently from each other, there is no communication
and no synchronization overhead.
Therefore, one may hope that with $p$ threads,
the speedup will be close to~$p$.

The number of paths that needs to be tracked to compute a numerical
irreducible decomposition can be a multiple of the number of paths
defined by a homotopy to approximate all isolated solutions.
Nevertheless, in order to properly distinguish the isolated singular
solutions (which occur with multiplicity two or higher) from the
solutions on positive dimensional solutions, one needs a representation
for the positive dimensional solution sets.

On parallel shared memory computers, the work crew model is applied.
In this model, threads are collaborating to complete a queue of jobs.
The pointer to the next job in the queue is guarded by a semaphore
so only one thread can access the next job and move the pointer to
the next job forwards.
The design of multithreaded software is described in~\cite{San11}.

The development of the blackbox solver was targeted at
the cyclic $n$-roots systems.  Backelin's Lemma~\cite{Bac89}
states that, if $n$ has a quadratic divisor, then there
are infinitely many cyclic $n$-roots.
Interesting values for $n$ are thus 8, 9, and 12,
respectively considered in~\cite{BF94}, \cite{Fau01}, and~\cite{Sab11}.
% For $n=16$, the tropical prevariety was computed in
% \cite{AV13,Bac89,BF94,BS95,DKK03,Fau01,FR15,Haa07,Sab11,SV16a}.

\noindent {\bf Problem Statement.}
The top down computation of a numerical irreducible decomposition
requires first the solving of a system augmented with as many
general linear equations as the expected top dimension of the
solution set.  This first stage is then followed by a cascade of
homotopies to compute candidate generic points on lower dimensional
solution sets.  In the third stage, the output of the cascades is
filtered and generic points are classified along their
irreducible components.  In the application of the work crew model
with $p$ threads, the problem is to study if the speedup will
converge to $p$, asymptotically for sufficiently large problems.
Another interesting question concerns {\em quality up}:
if we can afford the same computational time as on one thread,
then by how much can we improve the quality of the computed results
with $p$ threads?

\noindent {\bf Prior Work.}
The software used in this paper is PHCpack~\cite{Ver99},
which provides a numerical irreducible decomposition~\cite{SVW03}.
For the mixed volume computation,
MixedVol~\cite{GLW05} and DEMiCs~\cite{MT08} are used.
An introduction to the homotopy continuation methods for
computing positive dimensional solution sets is described in~\cite{SVW05}.
The overhead of double double and quad double precision~\cite{HLB00}
in path trackers can be compensated on multicore workstations
by parallel algorithms~\cite{VY10}.
The factorization of a pure dimensional solution set
on a distributed memory computer with message passing
was described in~\cite{LV09}.

\noindent {\bf Related Work.}
A numerical irreducible decomposition can be computed by a program
described in~\cite{BHSW08}, but that program lacks 
polyhedral homotopies, needed to efficiently solve sparse 
polynomial systems such as the cyclic $n$-roots problems.
Parallel algorithms for mixed volumes and polyhedral homotopies
were presented in~\cite{CLL14b,CLL14a}.
The computation of the positive dimensional solutions
for the cyclic 12-roots problem was reported first in~\cite{Sab11}.
A recent parallel implementation of polyhedral homotopies
was announced in~\cite{Mal15}.

\noindent {\bf Contributions and Organization.}
The next section proposes the application of pipelining to
interleave the computation of mixed cells with the tracking of
solution paths to solve a random coefficient system.
The production rate of mixed cells relative to the cost of
path tracking is related to the pipeline latency.
The third section describes the second stage in the solver
and examines the speedup for tracking paths defined by
sequences of homotopies.  In section four, the speedup of
the application of the homotopy membership test is defined.
One outcome of this research is free and open software
to compute a numerical irreducible decomposition on parallel
shared memory computers.  Computational experiments with
the software are presented in section five.

\section{Solving the Top Dimensional System}

There is only one input to the blackbox solver:
the expected top dimension of the solution set.
This input may be replaced by the number of variables minus one.
However, entering an expected top dimension that is too high
may lead to a significant computational overhead.

\subsection{Random Hyperplanes and Slack Variables}

A system is called {\em square} if it has as many equations
as unknowns.  A system is {\em underdetermined} if it has 
fewer equations than unknowns.  An underdetermined system 
can be turned into square system by adding as many linear
equations with randomly generated complex coefficients as
the difference between the number of unknowns and equations.
A system is {\em overdetermined} if there are more equations
than unknowns.  To turn an overdetermined system into a square one,
add repeatedly to every equation in the overdetermined system 
a random complex constant multiplied by a new slack variable, repeatedly
until the total number of variables equals the number of equations.

The top dimensional system is the given polynomial system,
augmented with as many linear equations with randomly generated
complex coefficients as the expected top dimension.
To the augmented system as many slack variables are added as
the expected top dimension.  The result of adding random linear
equations and slack variables is called an {\em embedded} system.
Solutions of the embedded system with zero slack variables are
generic points on the top dimensional solution set.
Solutions of the embedded system with nonzero slack variables
are start solutions in cascades of homotopies to compute
generic points on lower dimensional solution sets.

\begin{example} (embedding a system) \label{exembedding}
{\rm The equations for the cyclic 4-roots problem are
\begin{equation} \label{eqcyclic4}
   \bff(\x) =
   \left\{
      \begin{array}{c}
         x_1 + x_2 + x_3 + x_4 = 0 \\
         x_1 x_2 + x_2 x_3 + x_3 x_4 + x_4 x_1 = 0 \\
         x_1 x_2 x_3 + x_2 x_3 x_4 + x_3 x_4 x_1 + x_4 x_1 x_2 = 0 \\
         x_1 x_2 x_3 x_4 - 1 = 0.
      \end{array}
   \right.
\end{equation}
The expected top dimension equals one.
The system is augmented by one linear equation
and one slack variable~$z_1$.
The embedded system is then the following:
\begin{equation} \label{eqembeddedcyclic4}
   E_1(\bff(\x),z_1) =
   \left\{
      \begin{array}{rcl}
         x_1 + x_2 + x_3 + x_4 + \gamma_1 z_1 & = & 0 \\
         x_1 x_2 + x_2 x_3 + x_3 x_4 + x_4 x_1 + \gamma_2 z_1 & = & 0 \\
         x_1 x_2 x_3 + x_2 x_3 x_4 + x_3 x_4 x_1 + x_4 x_1 x_2
       + \gamma_3 z_1 & = & 0 \\
         x_1 x_2 x_3 x_4 - 1 + \gamma_4 z_1 & = & 0 \\
         c_0 + c_1 x_1 + c_2 x_2 + c_3 x_3 + c_4 x_4 + z_1 & = & 0.
      \end{array}
   \right.
\end{equation}
The constants $\gamma_1$, $\gamma_2$, $\gamma_3$, $\gamma_4$
and $c_0$, $c_1$, $c_2$, $c_3$, $c_4$ are randomly generated
complex numbers. 

The system $E_1(\bff(\x),z_1) = \zero$ has 20 solutions.
Four of those 20 solutions have a zero value for the slack variable~$z_1$.
Those four solutions satisfy thus the system
\begin{equation} \label{eqaugmentedcyclic4}
   E_1(\bff(\x),0) =
   \left\{
      \begin{array}{c}
         x_1 + x_2 + x_3 + x_4 = 0 \\
         x_1 x_2 + x_2 x_3 + x_3 x_4 + x_4 x_1 = 0 \\
         x_1 x_2 x_3 + x_2 x_3 x_4 + x_3 x_4 x_1 + x_4 x_1 x_2 = 0 \\
         x_1 x_2 x_3 x_4 - 1 = 0 \\
         c_0 + c_1 x_1 + c_2 x_2 + c_3 x_3 + c_4 x_4 = 0.
      \end{array}
   \right.
\end{equation}
By the random choice of the constants $c_0$, $c_1$, $c_2$,
$c_3$, and $c_4$, the four solutions are generic points
on the one dimensional solution set.
Four equals the degree of the one dimensional solution set
of the cyclic 4-roots problem.  }
\end{example}

For systems with sufficiently general coefficients,
polyhedral homotopies are generically optimal in the sense
that no solution path diverges.
Therefore, the default choice to solve the top dimensional system
is the computation of a mixed cell configuration and the solving of
a random coefficient start system.
Tracking the paths to solve the random coefficient start system is
a pleasingly parallel computation, which with dynamic load balancing
will lead to a close to optimal speedup.

\subsection{Pipelined Polyhedral Homotopies}

The computation of all mixed cells is harder to run in parallel,
but fortunately the mixed volume computation takes in general less time
than the tracking of all solution paths and, more importantly, the
mixed cells are not obtained all at once at the end,
but are produced in sequence, one after the other.
As soon as a cell is available, the tracking of as many solution paths
as the volume of the cell can start.
Figure~\ref{fig2stagepipe} illustrates a 2-stage pipeline
with $p$ threads.

\begin{figure}[hbt]
\begin{center}
\begin{picture}(220,55)(0,0)
\put(80,32){\line(1,0){15}}
\put(80,20){\line(1,0){15}}   \put(83,23){$P_0$}
\put(80,20){\line(0,1){12}}
\put(95,20){\line(0,1){12}}
\put(95,26){\vector(1,0){35}}
\put(130,0){\line(1,0){25}}  \put(133,3){$P_{p-1}$}
\put(130,12){\line(1,0){25}} \put(142,17){$\vdots$}
\put(130,31){\line(1,0){25}} \put(138,34){$P_2$}
\put(130,43){\line(1,0){25}} \put(138,46){$P_1$}
\put(130,55){\line(1,0){25}}
\put(130,0){\line(0,1){55}}
\put(155,0){\line(0,1){55}}
\put(45,26){\vector(1,0){35}}   \put(3,24){$\bfg(\x) = \zero$}
\put(155,26){\vector(1,0){35}}  \put(195,24){$\bfg^{-1}(\zero)$}
\end{picture}
\caption{A 2-stage pipeline with thread $P_0$ in the first stage
to compute the cells to solve the start systems with paths to be
tracked in the second stage by $p-1$ threads $P_1$, $P_2$,
$\ldots$, $P_{p-1}$.
The input to the pipeline is a random coefficient system $\bfg(\x) = \zero$
and the output are its solutions in the set $\bfg^{-1}(\zero)$.  }
\label{fig2stagepipe}
\end{center}
\end{figure}
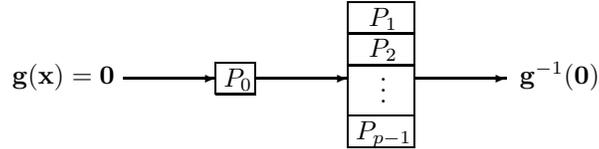

Figure~\ref{figpipediagram} illustrates the application of pipelining
to the solving of a random coefficient system where the subdivision
of the Newton polytopes has six cells.  The six cells are computed by
the first thread.  The other three threads take the cells and
run polyhedral homotopies to compute as many solutions as the volume
of the corresponding cell.

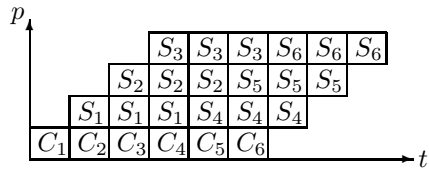
\begin{figure}[hbt]
\begin{center}
\begin{picture}(150,55)(0,0)
\put(0,0){\vector(0,1){53}}  \put(-7,53){$p$}
\put(0,0){\vector(1,0){145}} \put(147,-3){$t$}
\put(3,3){$C_1$}
\put(18,3){$C_2$}  \put(18,15){$S_1$}
\put(33,3){$C_3$}  \put(33,15){$S_1$}  \put(33,27){$S_2$}
\put(48,3){$C_4$}  \put(48,15){$S_1$}  \put(48,27){$S_2$}  \put(48,39){$S_3$}
\put(63,3){$C_5$}  \put(63,15){$S_4$}  \put(63,27){$S_2$}  \put(63,39){$S_3$}
\put(78,3){$C_6$}  \put(78,15){$S_4$}  \put(78,27){$S_5$}  \put(78,39){$S_3$}
                   \put(93,15){$S_4$}  \put(93,27){$S_5$}  \put(93,39){$S_6$}
                                      \put(108,27){$S_5$} \put(108,39){$S_6$}
                                                          \put(123,39){$S_6$}
\put(0,12){\line(1,0){105}}
\put(15,24){\line(1,0){105}}
\put(30,36){\line(1,0){105}}
\put(45,48){\line(1,0){90}}
\put(15,0){\line(0,1){24}}
\put(30,0){\line(0,1){36}}
\put(45,0){\line(0,1){48}}
\put(60,0){\line(0,1){48}}
\put(75,0){\line(0,1){48}}
\put(90,0){\line(0,1){48}}
\put(105,12){\line(0,1){36}}
\put(120,24){\line(0,1){24}}
\put(135,36){\line(0,1){12}}
\end{picture}
\caption{A space time diagram for a 2-stage pipeline with one thread
to produce 6 cells $C_1$, $C_2$, $\ldots$, $C_6$ and 3 threads 
to solve the corresponding 6 start systems $S_1$, $S_2$, $\ldots$, $S_6$.
For regularity, it is assumed that solving one start system takes three 
times as many time units as it takes to produce one cell. }
\label{figpipediagram}
\end{center}
\end{figure}

Counting the horizontal span of time units in Figure~\ref{figpipediagram},
the total time equals 9 units.  In the corresponding sequential process,
it takes 24 time units.  This particular pipeline with 4 threads
gives a speedup of $24/9 \approx 2.67$.

\subsection{Speedup}

As in Figure~\ref{fig2stagepipe},
consider a scenario with $p$ threads:
\begin{itemize}
\item the first thread produces $n$ cells; and
\item the other $p-1$ threads track all paths corresponding to the cells.
\end{itemize}
Assume that tracking all paths for one cell costs $F$ times the amount
of time it takes to produce that one cell.
In this scenario, the sequential time $T_1$,
the parallel time $T_p$, and the speedup $S_p$ are defined
by the following formulas:
\begin{equation} \label{eqpipespeedup}
  T_1 = n + Fn, \quad T_p = p-1 + \frac{Fn}{p-1}, \quad
  S_p = \frac{T_1}{T_p} = \frac{n (1+F)}{p-1 + \frac{Fn}{p-1}}.
\end{equation}
The term $p-1$ in $T_p$ is {\em the pipeline latency},
the time it takes to fill up the pipeline with jobs.
After this latency, the pipeline works at full speed.

The formula for the speedup $S_p$ in~(\ref{eqpipespeedup})
is rather too complicated for direct interpretation.
Let us consider a special case.
For large problems, the number $n$ of cells is larger
than the number $p$ of threads, $n \gg p$.
For a fixed number $p$ of threads, let $n$ approach infinity.
Then an optimal speedup is achieved, if the pipeline latency $p-1$
equals the multiplier factor $F$ in the tracking of all paths
relative to the time to produce one cell.
This observation is formalized in the following theorem.

\begin{theorem} \label{theopipe}
If $F=p-1$, then $S_p = p$ for $n \rightarrow \infty$.
\end{theorem}
\begin{proof}
For $F = p-1$, $T_1 = np$ and $T_p = n + p - 1$.  Then,
letting $n \rightarrow \infty$,
\begin{equation}
   \lim_{n \rightarrow \infty} S_p
   = \lim_{n \rightarrow \infty} \frac{T_1}{T_p}
   = \lim_{n \rightarrow \infty} \frac{np}{n+p-1} = p. \quad \qed
\end{equation}
\end{proof}

In case the multiplier factor is larger than the pipeline latency,
if $F > p-1$, then the first thread will finish sooner
with its production of cells and remains idle for some time.
If $p \gg 1$, then having one thread out of many idle is not bad.
The other case, if tracking all paths for one cell is smaller
than the pipeline latency, if $F < p-1$, is worse
as many threads will be idle waiting for cells to process.

The above analysis applies to pipelined polyhedral homotopies
to solve a random coefficient system.
Consider the solving of the top dimensional system.

\begin{corollary}
Let $F$ be the multiplier factor in the cost of tracking the paths
to solve the start system, relative to the cost of computing the cells.
If the pipeline latency equals $F$,
then the speedup to solve the top dimensional system
with $p$ threads will asymptotically converge to $p$,
as the number of cells goes to infinity.
\end{corollary}
\begin{proof}
Solving the top dimensional system consists in two stages.
The first stage, solving a random coefficient system,
is covered by Theorem~\ref{theopipe}.
In the second stage, the solutions of the random coefficient system are
the start solutions in a homotopy to solve the top dimensional system.
This second stage is a pleasingly parallel computation as the paths
can be tracked independently from each and for which the speedup is
close to optimal for sufficiently large problems.~\qed
\end{proof}

\section{Computing Lower Dimensional Solution Sets}

The solution of the top dimensional system is an important first stage,
which leads to the top dimensional solution set,
provided the given dimension on input equals the top dimension.
This section describes the second stage in a numerical irreducible
decomposition: the computation of candidate generic points
on the lower dimensional solution sets.

\subsection{Cascades of Homotopies}

The solutions of an embedded system with nonzero slack variables
are regular solutions and serve as start solutions to compute
sufficiently many generic points on the lower dimensional solution sets.
The sufficiently many in the sentence above means that there will
be at least as many generic points as the degrees of the lower dimensional
solution sets.

\begin{example} \label{exmultistep}
(a system with a 3-stage cascade of homotopies) {\rm
Consider the following system:
\begin{equation} \label{eqmultistep}
   \bff(\x)
   = \left\{
      \begin{array}{l}
         (x_1-1)(x_1-2)(x_1-3)(x_1-4) = 0 \\
         (x_1-1)(x_2-1)(x_2-2)(x_2-3) = 0 \\
         (x_1-1)(x_1-2)(x_3-1)(x_3-2) = 0 \\
         (x_1-1)(x_2-1)(x_3-1)(x_4-1) = 0.
      \end{array}
   \right.
\end{equation}
In its factored form, the numerical irreducible decomposition
is apparent.  First, there is the three dimensional solution set
defined by $x_1 = 1$.  
Second, for $x_1 = 2$, observe that 
$x_2 = 1$ defines a two dimensional solution set
and four lines: $(2, 2, x_3, 1)$, $(2, 2, 1, x_4)$,
$(2, 3, 1, x_4)$, and $(2, 3, x_3, 1)$.
Third, for $x_1 = 3$, there are four lines:
$(3, 1, 1, x_4)$, $(3, 1, 2, x_4)$, $(3, 2, 1, x_4)$, $(3, 3, 1, x_4)$,
and two isolated points $(3, 2, 2, 1)$ and $(3, 3, 2, 1)$.
Fourth, for $x_1 = 4$, there are four lines:
$(4, 1, 1, x_4)$, $(4, 1, 2, x_4)$, $(4, 2, 1, x_4)$,
$(4, 3, 1, x_4)$, and two additional isolated solutions
$(4, 3, 2, 1)$ and $(4, 2, 2, 1)$.

Sorted then by dimension, there is one three dimensional solution set,
one two dimensional solution set, twelve lines,
and four isolated solutions.

The top dimensional system has three random linear equations
and three slack variables $z_1$, $z_2$, and $z_3$.
The mixed volume of the top dimensional system equals 61
and this is the number of paths tracked in its solution.
Of those 61 paths, 6 diverge to infinity and the cascade
of homotopies starts with 55 paths.
The number of paths tracked in the cascade is summarized
at the right in Figure~\ref{figcascade}.

\begin{figure}[hbt]
\begin{center}
\begin{picture}(156,156)(0,0)
% level three, after solving the top
\put(3,147){\framebox{55}}
\put(150,147){\Ovalbox{55}} \put(159,144){\vector(0,-1){30}}
\put(19.5,150){\vector(1,0){20}}  
\put(40,120){\line(0,1){36}}
\put(86,120){\line(0,1){36}}
\put(40,156){\line(1,0){46}}
\put(40,144){\line(1,0){46}}  \put(43,147){6 at $\infty$}
\put(40,132){\line(1,0){46}}  \put(43,135){1 $z_3 = 0$}
\put(86,138){\vector(1,0){20}}
\put(106,135){\framebox{1}}
\put(40,120){\line(1,0){46}}  \put(43,123){54 $z_3 \not= 0$}
\put(11,125){\line(1,0){29}}
\put(11,125){\vector(0,-1){12}}
% level two, to the two dimensional components
\put(3,103){\framebox{54}}
\put(150,103){\Ovalbox{54}} \put(159,100){\vector(0,-1){30}}
\put(19.5,106){\vector(1,0){20}}
\put(86,76){\line(0,1){36}}
\put(40,76){\line(0,1){36}}
\put(40,112){\line(1,0){46}}
\put(40,100){\line(1,0){46}} \put(43,103){2 at $\infty$}
\put(40,88){\line(1,0){46}}  \put(43,91){2 $z_2 = 0$}
% filter the 2 candidate generic points
\put(86,94){\vector(1,0){20}}
\put(106,91){\framebox{2}}
% \put(116,94){\vector(1,0){20}}
% \put(136,91){\framebox{1}}
\put(40,76){\line(1,0){46}}  \put(43,79){50 $z_2 \not= 0$}
\put(11,81){\line(1,0){29}}
\put(11,81){\vector(0,-1){12}}
% level one, to the points at dimension one
\put(3,59){\framebox{50}}    
\put(150,59){\Ovalbox{54}}   \put(159,56){\vector(0,-1){30}}
\put(19.5,62){\vector(1,0){20}}
\put(86,32){\line(0,1){36}}
\put(40,32){\line(0,1){36}}
\put(40,68){\line(1,0){46}}
\put(40,56){\line(1,0){46}}  \put(43,59){6 at $\infty$}
\put(40,44){\line(1,0){46}}  \put(43,47){18 $z_1 = 0$}
% filter the 18 candidate generic points
\put(86,50){\vector(1,0){20}}
\put(106,47){\framebox{18}}
% \put(120,50){\vector(1,0){20}}
% \put(140,47){\framebox{15}}
% \put(156,50){\vector(1,0){20}}
% \put(176,47){\framebox{12}}
\put(40,32){\line(1,0){46}}  \put(43,35){26 $z_1 \not= 0$}
\put(11,37){\line(1,0){29}} 
\put(11,37){\vector(0,-1){12}}
% level zero, to the isolated solutions
\put(3,15){\framebox{26}}    
\put(150,15){\Ovalbox{26}} 
\put(19.5,18){\vector(1,0){20}}
\put(40,12){\line(1,0){46}}  \put(43,15){18 at $\infty$}
\put(40,0){\line(1,0){46}}   \put(43,3){8 finite}
% filter the 8 candidate isolated solutions
\put(40,24){\line(1,0){46}}
\put(40,0){\line(0,1){24}}
\put(86,0){\line(0,1){24}}
\put(86,6){\vector(1,0){20}}
\put(106,3){\framebox{8}}
% \put(116,6){\vector(1,0){20}}
% \put(136,3){\framebox{5}}
% \put(148,6){\vector(1,0){20}}
% \put(168,3){\framebox{5}}
% \put(180,6){\vector(1,0){20}}
% \put(200,3){\framebox{4}}
\end{picture}
\caption{At the left are the numbers of paths tracked
in each stage of the computation of a numerical irreducible 
decomposition of $\bff(\x) = \zero$ in~(\ref{eqmultistep}).
The numbers at the right are the {\em candidate} generic points on
each positive dimensional solution set,
or in case of the rightmost 8 at the bottom,
the number of {\em candidate} isolated solutions.
Shown at the farthest right is the summary of the number of paths
tracked in each stage of the cascade.}
\label{figcascade}
\end{center}
\end{figure}
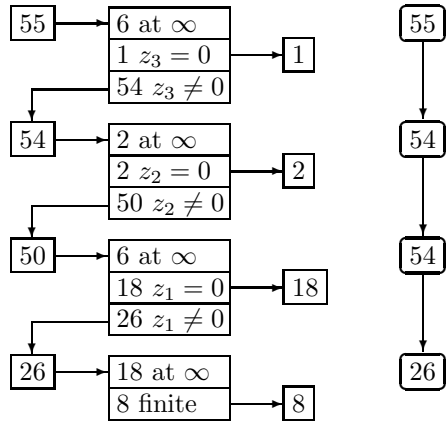

The number of solutions with nonzero slack variables remains
constant in each run, because those solutions are regular.
Except for the top dimensional system, the number of solutions
with slack variables equal to zero fluctuates,
each time different random constants are generated in the embedding,
because such solutions are highly singular. }
\end{example}

The right of Figure~\ref{figcascade} shows the order of computation 
of the path tracking jobs, in four stages, for each dimension of
the solution set.  The obvious parallel implementation is to have
$p$ threads collaborate to track all paths in that stage.

\subsection{Speedup}

The following analysis assumes that every path has the same
difficulty and requires the same amount of time to track.
\begin{theorem} \label{theopathspeedup}
Let $T_p$ be the time it takes to track $n$ paths with $p$ threads.
Then, the optimal speedup $S_p$ is
\begin{equation}
   S_p = p - \frac{p-r}{T_p}, \quad r = n \mbox{\rm ~mod~} p.
\end{equation}
If $n < p$, then $S_p = n$.
\end{theorem}
\begin{proof}
Assume it takes one time unit to track one path.
The time on one thread is then $T_1 = n = q p + r$,
$q = \lfloor n/p \rfloor$ and $r = n \mbox{\rm ~mod~} p$.
As $r < p$, the tracking of $r$ paths with $p$ threads 
takes one time unit, so $T_p = q + 1$.  Then the speedup is
\begin{equation}
   S_p = \frac{T_1}{T_p}
       = \frac{q p + r}{q + 1}
       = \frac{q p + p - p + r}{q + 1}
       = \frac{q p + p}{q + 1} - \frac{p-r}{q+1}
       = p - \frac{p-r}{T_p}.
\end{equation}
If $n < p$, then $q=0$ and $r = n$, which leads to~$S_p = n$. \hfill \qed
\end{proof}

In the limit, as $n \rightarrow \infty$, also $T_p \rightarrow \infty$,
then $(p-r)/T_p \rightarrow 0$ and so $S_p \rightarrow p$.
For a cascade with $D+1$ stages, Theorem~\ref{theopathspeedup} can
be generalized as follows.

\begin{corollary} \label{corpathspeedup}
Let $T_p$ be the time it takes to track with $p$ threads
a sequence of $n_0$, $n_1$, $\ldots$, $n_D$ paths.
Then, the optimal speedup $S_p$ is
\begin{equation} \label{eqpathspeedup}
   S_p = p - \frac{dp-r_0 - r_1 - \cdots - r_D}{T_p}, 
   \quad r_k = n_k \mbox{\rm ~mod~} p, k = 0, 1, \ldots D.
\end{equation}
\end{corollary}
\begin{proof}
Assume it takes one time unit to track one path.
The time on one thread is then 
\begin{equation}
   T_1 = n_0 + n_1 + \cdots + n_D
       = q_0 p + r_0 + q_1 p + r_1 + \cdots + q_D p + r_D,
\end{equation}
where $q_k = \lfloor n_k/p \rfloor$ and $r_k = n_k \mbox{\rm ~mod~} p$,
for $k=0,1,\ldots,D$.
As $r_k < p$, the tracking of $r_k$ paths with $p$ threads 
takes $D+1$ time units, so the time on $p$ threads is
\begin{equation}
   T_p = q_0 + q_1 + \cdots + q_D + D+1.
\end{equation}
Then the speedup is
\begin{eqnarray}
   S_p = \frac{T_1}{T_p}
       & = & \frac{p T_p - dp + r_0 + r_1 + \cdots + r_D}{T_p} \\
       & = & p - \frac{dp - r_0 - r_1 - \cdots - r_D}{T_p}. \quad \qed
\end{eqnarray}
\end{proof}

If the length $D+1$ of the sequence of paths is long
and the number of paths in each stage is less than $p$,
then the speedup will be limited.

\section{Filtering Lower Dimensional Solution Sets}

Even if one is interested only in the isolated solutions
of a polynomial system, one would need to be able to distinguish the
isolated multiple solutions from solutions on a positive dimensional
solution set.  Without additional information, both an isolated
multiple solution and a solution on a positive dimensional set
appear numerically as singular solutions, that is: as solutions
where the Jacobian matrix does not have full rank.
A homotopy membership test makes this distinction.

\subsection{Homotopy Membership Tests}

\begin{example} \label{exmembtest}
(homotopy membership test) {\rm
Consider the following system:
\begin{equation} \label{eqmbthomsys}
   \bff(\x) =
   \left\{
      \begin{array}{rcl}
         (x_1 - 1)(x_1 - 2) & = &  0 \\
            (x_1 - 1) x_2^2 & = & 0. \\
      \end{array}
   \right.
\end{equation}
The solution consists of the line $x_1 = 1$
and the isolated point $(2, 0)$ which occurs
with multiplicity two.
The line $x_1 = 1$ is represented by one generic point
as the solution of the embedded system
\begin{equation} \label{eqmbthomemb}
   E(\bff(\x),z_1) =
   \left\{
      \begin{array}{rcl}
         (x_1 - 1)(x_1 - 2) + \gamma_1 z_1 & = & 0 \\
            (x_1 - 1) x_2^2 + \gamma_2 z_1 & = & 0 \\
             c_0 + c_1 x_1 + c_2 x_2 + z_1 & = & 0,
      \end{array}
   \right.
\end{equation}
where the constants $\gamma_1$, $\gamma_2$,
$c_0$, $c_1$, and $c_2$ are randomly generated complex numbers.
Replacing the constant $c_0$ by $c_3 = - 2 c_1$ makes that
the point $(2, 0, 0)$ satisfies the system $E(\bff(\x),z_1) = \zero$.
Consider the homotopy 
\begin{equation} \label{eqmbthomotopy}
   \bfh(\x,z_1,t) = 
   \left\{
      \begin{array}{rcl}
         (x_1 - 1)(x_1 - 2) + \gamma_1 z_1 & = & 0 \\
            (x_1 - 1) x_2^2 + \gamma_2 z_1 & = & 0 \\
         (1-t) c_0 + tc_3 + c_1 x_1 + c_2 x_2 + z_1 & = & 0.
      \end{array}
   \right.
\end{equation}
For $t=0$, there is the generic point on the line $x_1 = 1$
as a solution of the system~(\ref{eqmbthomemb}).
Tracking one path starting at the generic point to $t=1$
moves the generic point to another generic point on $x_1 = 1$.
If that other generic point at $t=1$ coincides with the
point $(2,0,0)$, then the point $(2,0)$ belongs to the line.
Otherwise, as is the case in this example, it does not.  }
\end{example}

In running the homotopy membership test,
a number of paths need to be tracked.
To identify the bottlenecks in a parallel version,
consider the output of Figure~\ref{figcascade}
in the continuation of the example on the system in~\ref{eqmultistep}.

\begin{example}  \label{exmultistep2}
(Example~\ref{exmultistep} continued) {\rm
Assume the spurious points on the higher dimensional solution sets
have already been removed so there is one generic point on
the three dimensional solution set, one generic point on 
the two dimensional solution set, and twelve generic points on
the one dimensional solution set.

At the end of the cascade, there are eight candidate isolated solutions.
Four of those eight are regular solutions and are thus isolated.
The other four solutions are singular.
Singular solutions may be isolated multiple solutions, 
but could also belong to the higher dimensional solution sets.
Consider Figure~\ref{figisolatedfilter}.

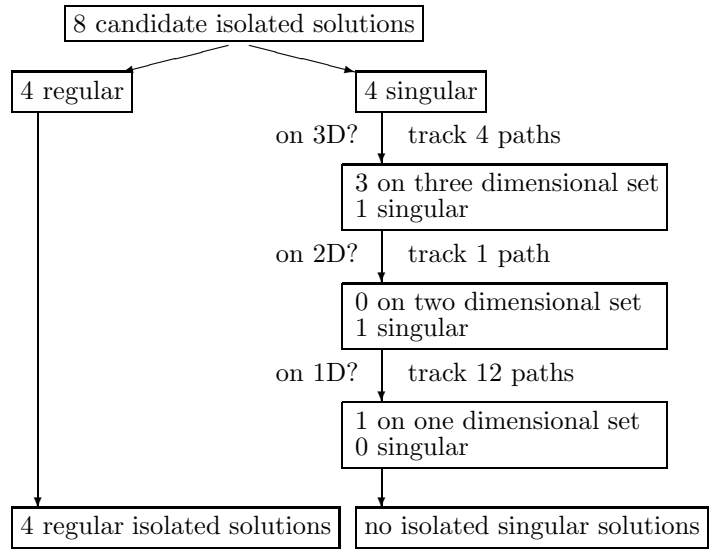
\begin{figure}[hbt]
\begin{center}
\begin{picture}(250,200)(0,0)
\put(20,190){\framebox{8 candidate isolated solutions}}
\put(82,185){\vector(-4,-1){40}}
\put(0,165){\framebox{4 regular}}
\put(10,160){\vector(0,-1){150}}
\put(0,0){\framebox{4 regular isolated solutions}}
\put(90,185){\vector(+4,-1){40}}
\put(130,165){\framebox{4 singular}}
\put(140,160){\vector(0,-1){20}}
\put(100,148){on 3D?}
\put(150,148){track 4 paths}
\put(125,140){\line(1,0){123}} \put(125,140){\line(0,-1){25}}
\put(125,115){\line(1,0){123}} \put(248,115){\line(0,1){25}}
\put(130,130){3 on three dimensional set}
\put(130,120){1 singular}
\put(140,115){\vector(0,-1){20}}
\put(100,103){on 2D?}
\put(150,103){track 1 path}
\put(125,95){\line(1,0){123}} \put(125,95){\line(0,-1){25}}
\put(125,70){\line(1,0){123}} \put(248,70){\line(0,1){25}}
\put(130,85){0 on two dimensional set}
\put(130,75){1 singular}
\put(140,70){\vector(0,-1){20}}
\put(100,58){on 1D?}
\put(150,58){track 12 paths}
\put(125,50){\line(1,0){123}} \put(125,50){\line(0,-1){25}}
\put(125,25){\line(1,0){123}} \put(248,25){\line(0,1){25}}
\put(130,40){1 on one dimensional set}
\put(130,30){0 singular}
\put(140,25){\vector(0,-1){15}}
\put(130,0){\framebox{no isolated singular solutions}}
\end{picture}
\caption{Stages in testing whether the singular candidate isolated points
belong to the higher dimensional solution sets.}
\label{figisolatedfilter}
\end{center}
\end{figure}

\noindent Executing the homotopy membership tests as in
in Figure~\ref{figisolatedfilter}, 
first on 3D, then on 2D, and finally on 1D,
the bottleneck occurs in the middle,
where there is only one path to track.  }
\end{example}

Figure~\ref{figcascadefilter} is the continuation
of Figure~\ref{figcascade}: the output of the cascade shown
in Figure~\ref{figcascade} is the input of the filtering
in Figure~\ref{figcascadefilter}.
Figure~\ref{figisolatedfilter} explains the last stage
in Figure~\ref{figcascadefilter}.

\begin{figure}[hbt]
\begin{center}
\begin{picture}(300,100)(84,0)
\put(84,90){\vector(0,-1){96}}
% level three, after solving the top
\put(84,78){\vector(1,0){20}}
\put(104,75){\framebox{1}}
% level two, to the two dimensional components
% filter the 2 candidate generic points
\put(260,66){\Ovalbox{1}} \put(266,62){\vector(0,-1){15}}
\put(84,54){\vector(1,0){20}}
\put(104,51){\framebox{2}}
\put(116,54){\vector(1,0){20}}
\put(136,51){\framebox{1}}
% level one, to the points at dimension one
% filter the 18 candidate generic points
\put(260,37){\Ovalbox{6}} \put(272,40){\vector(1,0){20}}
\put(292,37){\Ovalbox{3}} \put(298,33){\vector(0,-1){15}}
\put(84,30){\vector(1,0){20}}
\put(104,27){\framebox{18}}
\put(120,30){\vector(1,0){20}}
\put(140,27){\framebox{15}}
\put(156,30){\vector(1,0){20}}
\put(176,27){\framebox{12}}
% level zero, to the isolated solutions
% filter the 8 candidate isolated solutions
\put(292,8){\Ovalbox{4}} \put(304,11){\vector(1,0){20}}
\put(324,8){\Ovalbox{1}} \put(336,11){\vector(1,0){20}}
\put(356,8){\Ovalbox{12}}
\put(84,6){\vector(1,0){20}}
\put(104,3){\framebox{8}}
\put(116,6){\vector(1,0){20}}
\put(136,3){\framebox{5}}
\put(148,6){\vector(1,0){20}}
\put(168,3){\framebox{5}}
\put(180,6){\vector(1,0){20}}
\put(200,3){\framebox{4}}
\end{picture}
\caption{On input are the candidate generic points 
shown as output in Figure~\ref{figcascade}:
1 point at dimension three,
2 points at dimension two,
18 points at dimension one, and 
8 candidate isolated points.
Points on higher dimensional solution sets are removed
by homotopy membership filters.
The numbers at the right equal the number of paths
in each stage of the filters.
The sequence 4, 1, 12 at the bottom is explained
in Figure~\ref{figisolatedfilter}.  }
\label{figcascadefilter}
\end{center}
\end{figure}
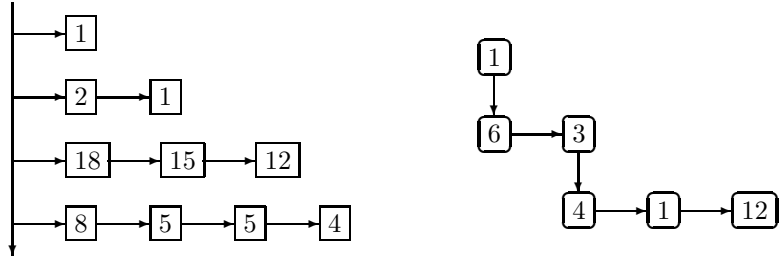

\subsection{Speedup}

The analysis of the speedup is another consequence
of Theorem~\ref{theopathspeedup}.

\begin{corollary} \label{corfilterspeedup}
Let $T_p$ be the time it takes to filter
$n_D$, $n_{D-1}$, $\ldots$, $n_{\ell+1}$ singular points
on components respectively of dimensions $D$, $D-1$, $\ldots$, $\ell+1$
and degrees $d_D$, $d_{D-1}$, $\ldots$, $d_{\ell+1}$.
Then, the optimal speedup is
\begin{equation} \label{eqfilterspeedup}
   S_p = p - \frac{(D-\ell)p - r_D - r_{D-1} - \cdots - r_{\ell+1}}{T_p},
   \quad r_k = (n_k d_k) \mbox{\rm ~mod~} p,
\end{equation}
for $k=\ell+1,\ldots,D-1, D$.
\end{corollary}
\begin{proof}
For a component of degree $d_k$,
it takes $n_k d_k$ paths to filter $n_k$ singular points.
The statement in~(\ref{eqfilterspeedup})
follows from replacing $n_k$ by $n_k d_k$
in the statement in~(\ref{eqpathspeedup}) 
of Corollary~\ref{corpathspeedup}. \hfill \qed
\end{proof}

Although the example shown in Figure~\ref{figcascadefilter}
is too small for parallel computation, it illustrates the
law of diminishing returns in introducing parallelisms.
There are two reasons for a reduced parallelism:
\begin{enumerate}
\item The number of singular solutions and the degrees of the
      solution sets could be smaller than the number of available cores.
\item In a cascade of homotopies, there are as many steps as $D+1$,
      where $D$ is the expected top dimension.
      To filter the output of the cascade, there are $D(D+1)/2$ stages,
      so longer sequences of homotopies are considered.
\end{enumerate}

Singular solutions that do not lie on any higher positive dimensional
solution set need to be processed further by 
deflation~\cite{LVZ06,LVZ07}, not available yet in
% Because the current implementation of deflation requires 
% dynamic memory allocation for storing representations 
% of the higher order derivatives,
a multithreaded implementation.
Parallel algorithms to factor the positive dimensional solutions
into irreducible factors are described in~\cite{LV09}.

\section{Computational Experiments}

% This section describes preliminary timings with the software.

% \subsection{Equipment and Software}

The software was developed on a Mac OS X laptop
and Linux workstations.
The executable for Windows also supports multithreading.
% The MacBook Pro runs runs macOS Sierra Version 10.12.6
% on a 3.1 GHz Intel Core i7 with 2 cores, 256 KB L2 cache per core,
% 4 MB L3 cache and 16 GB memory,
% in two banks of 8 GB running at 1867 MHz.
All times reported below are on a CentOS Linux 7 computer
with two Intel Xeon E5-2699v4 Broadwell-EP 2.20 GHz processors,
which each have 22 cores, 256 KB L2 cache and 55 MB L3 cache.
The memory is 256 MB, in 8 banks of 32 MB at 2400 MHz.
As the processors support hyperthreading, speedups of more than~44
are possible.

On Linux, the executable {\tt phc} is compiled with
the GNAT GPL 2016 edition of the gnu-ada compiler.
The thread model is posix, in gcc version~4.9.4.
The code in PHCpack
contains an Ada translation of the MixedVol Algorithm~\cite{GLW05},
The source code for the software is at github,
licensed under GNU GPL version 3.
The blackbox solver for a numerical irreducible decomposition
is called as {\tt phc -B} and with {\tt p} threads: as~{\tt phc -B -tp}.
With {\tt phc -B2} and {\tt phc -B4}, computations happen 
respectively in double double and quad double arithmetic~\cite{HLB00}.

\subsection{Solving Cyclic 8 and Cyclic 9-Roots} 

Both cyclic 8 and cyclic 9-roots are relatively small problems,
relative compared to the cyclic 12-roots problem.
Table~\ref{tabcyc89times} summarizes wall clock times and speedups 
for runs on the cyclic 8 and 9-roots systems.
The wall clock time is the real time, elapsed since the start
and the end of each run.  This includes the CPU time, system time,
and is also influenced by other jobs the operating system is running.

\begin{table}[hbt]
\begin{center}
\begin{tabular}{r||r|r||r|r}
    & \multicolumn{2}{c||}{cyclic 8-roots}
    & \multicolumn{2}{c}{cyclic 9-roots} \\
 p~ & ~seconds~ &  ~speedup~ &   seconds~  &  ~speedup~ \\ \hline \hline
 1~ & ~181.765~ &   1.00~~ &   ~2598.435~ &    1.00~~  \\
 2~ & 167.871~ &   1.08~~ &   1779.939~ &    1.46~~  \\
 4~ &  89.713~ &   2.03~~ &    901.424~ &    2.88~~  \\
 8~ &  47.644~ &   3.82~~ &    427.800~ &    6.07~~  \\
16~ &  32.215~ &   5.65~~ &    267.838~ &    9.70~~  \\
32~ &  22.182~ &   8.19~~ &    153.353~ &   16.94~~  \\
64~ &  20.103~ &   9.04~~ &    150.734~ &   17.24~~  \\
\end{tabular}
\caption{Wall clock times in seconds with {\tt phc -B -tp} for
{\tt p} threads.}
\label{tabcyc89times}
\end{center}
\end{table}

With 64 threads the time for cyclic 8-roots reduces
from 3 minutes to 20 seconds and for cyclic 9-roots
from 43 minutes to 2 minutes and 30 seconds.
Table~\ref{tabcyc89ddqdtimes} summarizes the wall clock times
with 64 threads in higher precision.

\begin{table}[hbt]
\begin{center}
\begin{tabular}{r||rcr||rcr}
    & \multicolumn{3}{c||}{cyclic 8-roots}
    & \multicolumn{3}{c}{cyclic 9-roots} \\
    & ~seconds & = & hms format~ & ~seconds  & = & hms format 
\\ \hline \hline
 dd~ &  53.042 & = &    53s~ &   498.805 & = &    8m19s~  \\
 qd~ & 916.020 & = & 15m16s~ & ~4761.258 & = & 1h19m21s~  \\
\end{tabular}
\caption{Wall clock times with 64 threads 
in double double and quad double precision.}
\label{tabcyc89ddqdtimes}
\end{center}
\end{table}

\subsection{Solving Cyclic 12-Roots on One Thread}

The classical B\'{e}zout bound for the system is 479,001,600.
This is lowered to 342,875,319 with the application of
a linear-product start system.  In contrast, the mixed volume
of the embedded cyclic 12-roots system equals 983,952.

The wall clock time on the blackbox solver on one thread is about
95 hours (almost 4 days).
This run includes the computation of the linear-product bound
which takes about 3 hours.  
This computation is excluded
in the parallel version because the multithreaded version overlaps
the mixed volume computation with polyhedral homotopies.
While a speedup of about 30 is not optimal, the time reduces
from 4 days to less than 3 hours with 64 threads,
see Table~\ref{tabpipeline}.

The blackbox solver does not exploit symmetry,
see~\cite{AV13} for such exploitation.

\subsection{Pipelined Polyhedral Homotopies}

This section concerns the computation of a random coefficient
start system used in a homotopy to solve the top dimensional system,
to start the cascade homotopies for the cyclic 12-roots system.
Table~\ref{tabpipeline} summarizes the wall clock times to
solve a random coefficient start system to solve the top dimensional system.
% system in the cascade for the cyclic 12-roots system.

\begin{table}[hbt]
%\begin{center}
{\small
\hspace{-5mm}~\begin{tabular}{r|rr|r|rr|c}
{\tt p}~ &  seconds  & hms format~ & ~speedup~
& ~total seconds  & hms format &  percentage \\
\hline \hline
 2~ & 62812.764 & 17h26m52s &    1.00~~
&    157517.816  & 43h45m18s &   39.88\% \\
 4~ & 21181.058  &  5h53m01s  &    2.97~~
&     73088.635  & 20h18m09s &   28.98\% \\
 8~ &  8932.512  &  2h28m53s  &    7.03~~
&     38384.005  & 10h39m44s &   23.27\% \\
16~ &  4656.478  &  1h17m36s  &   13.49~~
&     19657.329  &  5h27m37s &   23.69\% \\
32~ &  4200.362  &  1h10m01s  &   14.95~~
&     12154.088  &  3h22m34s &   34.56\% \\
64~ &  4422.220  &  1h13m42s  &   14.20~~
&      9808.424  &  2h43m28s &   45.08\%
\end{tabular}
}
%\begin{center}
\caption{Times of the pipelined polyhedral homotopies
versus the total time in the solver {\tt phc -B -tp},
for increasing values 2, 4, 8, 16, 32, 64 of the tasks {\tt p}.}
\label{tabpipeline}
%\end{center}
\end{table}

For pipelining, we need at least 2 tasks: 
one to produce the mixed cells and another to track the paths.
The speedup of {\tt p} tasks is computed over 2 tasks.
% The first line in Table~\ref{tabpipeline} is the time with
% one thread producing the mixed cells and another thread
% tracking the paths.  For the pipelined polyhedral homotopies,
% this is the reference to measure speedups.
With 16 threads, the time to solve a random coefficient system 
is reduced from 17.43 hours to 1.17 hour.
The second part of Table~\ref{tabpipeline} lists the time of solving 
the random coefficient system relative to the total time of the solver.
For 2 threads, solving the random coefficient system
takes almost 40\% of the total time and then decreases
to less than 24\% of the total time with 16 threads.
Already for 16 threads, the speedup of 13.49 indicates
that the production of mixed cells cannot keep up 
with the pace of tracking the paths.

Dynamic enumeration~\cite{MTK07} applies a greedy algorithm 
to compute all mixed cells and its implementation in DEMiCs~\cite{MT08}
produces the mixed cells at a faster pace than MixedVol~\cite{GLW05}.
Table~\ref{tabpipeline2} shows times for the mixed volume computation
with DEMiCs~\cite{MT08} in a pipelined version of the polyhedral homotopies.
% to solve the start system for the embedded cyclic 12-roots problem.

\begin{table}[hbt]
\begin{center}
\begin{tabular}{r|rcr|r}
{\tt p}~ &  ~seconds  & = & hms format~ & ~speedup~ \\
\hline \hline
 2~ & ~56614 & = & 15h43m34s~ &     1.00~~ \\
 4~ &  21224 & = &  5h53m44s~  &    2.67~~ \\
 8~ &   9182 & = &  2h23m44s~  &    6.17~~ \\
16~ &   4627 & = &  1h17m07s~  &   12.24~~ \\
32~ &   2171 & = &    36m11s~  &   26.08~~ \\
64~ &   1989 & = &    33m09s~  &   28.46~~
\end{tabular}
\caption{Times of the pipelined polyhedral homotopies with DEMiCs,
for increasing values 2, 4, 8, 16, 32, 64 of tasks {\tt p}.
The last time is an average over 13 runs.  With 64 threads the
times ranged between 23 minutes and 47 minutes.}
\label{tabpipeline2}
\end{center}
\end{table}

\subsection{Solving the Cyclic 12-Roots System in Parallel}

As already shown in Table~\ref{tabpipeline},
the total time with 2 threads goes down from more than 43 hours
to less than 3 hours, with 64 threads.
Table~\ref{tabcyclic12times} provides a detailed breakup
of the wall clock times for each stage in the solver.

\begin{table}[hbt]
%\begin{center}
{\small
\begin{tabular}{r||r|r|r||r|r|r||r|r}
   & \multicolumn{3}{c||}{solving top system}
   & \multicolumn{3}{c||}{cascade and filter}
   & \multicolumn{1}{c|}{grand} \\
 p~ 
   & \multicolumn{1}{c|}{start}
   & \multicolumn{1}{c|}{contin}
   & \multicolumn{1}{c||}{total}
   & \multicolumn{1}{c|}{~cascade~}
   & \multicolumn{1}{c|}{filter}
   & \multicolumn{1}{c||}{total} 
   & \multicolumn{1}{c|}{total}
 & ~speedup \\ 
\hline \hline
 2~ &  62813  & 47667 & 110803 
    &  44383  & 2331 & 46714 & 157518 &  1.00~~ \\
 4~ &  21181  &  25105 &   46617 
    &  24913  & 1558 & 26471 &  73089 &  2.16~~ \\
 8~ &   8933  &  14632 &   23896 
    &  13542  &  946 & 14488 &  38384 &  4.10~~ \\
16~ &   4656  &   7178 &   12129 
    &   6853  &  676 &  7529 &  19657 &  8.01~~ \\
32~ &   4200  &   3663 &    8094
    &   3415  &  645 &  4060 &  12154 & 12.96~~ \\ 
64~ &   4422  &   2240 &    7003
    &   2228  &  557 &  2805 &   9808 & 16.06~~ \\
\end{tabular}
}
\caption{Wall clock times in seconds for all stages
of the solver on cyclic 12-roots.
The solving of the top dimension system breaks up in two stages:
the solving of a start system (start) and the continuation 
to the solutions of the the top dimensional system (contin).
Speedups are good in the cascade stage, but the filter stage
contains also the factorization in irreducible components,
which does not run in parallel. }
\label{tabcyclic12times}
% \end{center}
\end{table}

A run in double double precision with 64 threads ends after
7 hours and 37 minutes.
This time lies between the times in double precision
with 8 threads, 10 hours and 39 minutes, and 
with 16 threads, 5 hours and 27 minutes (Table~\ref{tabpipeline}).
% In double double precision, tracking of paths is more expensive
% which improves the efficiency of the pipelined polyhedral homotopies,
% which finish after one hour and 47 minutes.
Confusing quality with precision, from 8 to 64 threads,
the working precision can be doubled with a reduction in time
by 3 hours, from 10.5 hours to 7.5 hours.

% \section{Conclusions}

%In conclusion, this paper describes a coarse grained parallel algorithm
%for shared memory computers.
%The speedups are sufficient to compensate for raising
%the working precision from double to double double,
%thus achieving quality up.
%Future improvements will include dynamic enumeration
%and GPU acceleration.

\bibliographystyle{plain}
% \bibliography{art,cyclic,related}

\end{document}